\documentclass{IEEEtran}

\usepackage{stfloats}

\usepackage{lipsum} 
\usepackage{fancyhdr}
\usepackage{mathtools}
\usepackage{amsmath}
\usepackage{amssymb}
\usepackage{amsmath}
\usepackage{cite}
\usepackage{siunitx}
\usepackage{mathtools}
\usepackage{amsmath}
\usepackage{amssymb}
\usepackage{amsmath}
\usepackage{cite}
\usepackage{siunitx}

\usepackage{amsthm}
\usepackage{graphicx}
\usepackage{graphicx}
\usepackage[utf8]{inputenc}

\usepackage{amsthm}

\usepackage{epstopdf}
\usepackage{mdwmath}
\usepackage{blindtext}
\usepackage{eqparbox}
\usepackage{fixltx2e}

\usepackage{graphicx}
\usepackage[export]{adjustbox}

\usepackage{eqparbox}

\newtheorem{theorem}{Theorem}
\newtheorem{lemma}{Lemma}

\begin{document}
\title{On Optimal Online Algorithms for Energy Harvesting Systems with Continuous Energy and Data Arrivals} %
\author{\normalsize Milad Rezaee,
        Mahtab~Mirmohseni and~Mohammad Reza Aref\\
        \thanks{M. Rezaee, M. Mirmohseni and M. R. Aref are with the Information Systems and Security Laboratory,
        Department of Electrical Engineering, Sharif University of Technology,
        Tehran 11365/8639, Iran (e-mail: miladrezaee@ee.sharif.edu; mirmohseni@
        sharif.edu; aref@sharif.edu). }
      \vspace{0ex}
       }

\markboth{\vspace{0ex}}%
{}
\maketitle

\begin{abstract}
\boldmath
Energy harvesting (EH) has been developed to extend the lifetimes of energy-limited communication systems. In this letter, we consider a single-user EH communication system, in which both of the arrival data and the harvested energy curves are modeled as \emph{general} functions. Unlike most of the works in the field, we investigate the online algorithms which only acquire the causal information of the arrival data and the harvested energy processes. We study how well the optimal online algorithm works compared with the optimal offline algorithm, and thus our goal is to find the lower and upper bounds for the ratio of the completion time in the optimal online algorithm to the optimal offline algorithm. We propose two online algorithms which achieves the upper bound of $2$ on this ratio. Also, we show that this ratio is $2$ for the optimal online algorithm.
\vspace{0ex}
\end{abstract}

\IEEEpeerreviewmaketitle
\section{Introduction}
\fontsize{9.9}{11} \selectfont
%
%
%
%
Providing the required energy from the natural renewable sources, Energy harvesting (EH) systems not only improve lifetime of the wireless systems, but also have been developed to make the green communications possible. Recently, technology progress has donated towards realizing effective practical design of EH devices, yielding the  sufficient power required for communications. EH systems differ from conventional communication systems in that a constant data transmission rate cannot be guaranteed because of sporadic nature of received energy as well as causal information about it.

One of the challenges in EH systems is finding a policy to optimize delay or throughput as the evaluation metric, which results in the throughput maximization or the completion time minimization. Both problems have been considered with the assumption of either noncausal or causal knowledge of the harvested energy process, where their corresponding algorithms are referred to as offline or online, respectively. The offline algorithms need the information about the future of the (harvested energy and/or arrival data) process and potentially perform better than the causal algorithms, which only require the past and current information.
The throughput maximization problem has been studied with both noncausal and causal knowledge of the EH process.
In \cite{tutuncuoglu2012optimum}, the optimal offline algorithm for the throughput maximization problem is investigated when the harvested energy is modeled as a discrete curve.
Considering a discrete model for the EH process in \cite{vaze2013dynamic}, an online algorithm to maximize the throughput of a wireless channel with arbitrary fading coefficients is designed. Recently, a few works have considered a continuous model for the harvested energy in the throughput maximization problem \cite{varan2014energy,milad,rezaee2016optimal}.

Our focus in this letter is on the delay metric, i.e., the \emph{time minimization problem}, which is investigated in \cite{yang2012optimal,vaze2013competitive,zheng2016online}.
To study the performance of an online algorithm one may consider the ratio of the completion time in the online algorithm to the optimal offline algorithm (called the online \emph{time efficiency ratio}). An online algorithm is called $\rho$-competitive if its online time efficiency ratio is less than or equal to $\rho$.
In \cite{yang2012optimal}, the optimal offline algorithm is proposed to minimize the completion time for the discrete model with an EH transmitter (Tx). In \cite{vaze2013competitive}, a lower bound and an upper bound on the optimal online time efficiency ratio are found in order to transfer a given data in a single-user and a multiple-access channels. It is proved in \cite{zheng2016online} that the optimal online algorithm is 2-competitive for the discrete model with an EH Tx. The only work that considers the data arrival process in the completion time minimization problem proposes the optimal \emph{offline} algorithm in the discrete setup \cite{yang2012optimal}.

Only requiring the causal information, the online algorithms are more practical than the offline algorithms while it is difficult to obtain their analytical performance limits. In fact, the challenge is to derive bounds on the online time efficiency ratio of a proposed online algorithm. In this paper, we consider the time minimization problem while the \emph{data arrival process} is taken into account. Our purpose is to investigate the optimal online time efficiency ratio in a single-user EH system in a \emph{continuous} model (noting that the results are applicable to the discrete model, too). In consistence with the prior models in \cite{vaze2013competitive,zheng2016online,vaze2013dynamic}, we assume that we have no access to the distributions of harvested energy and arrival data. We remark that our proposed online algorithms can be used where we know the distributions of harvested energy and arrival data as well. But considering distributions in online algorithms severely increases computational complexity. Our contributions has been listed as:
i) We propose two online algorithms and prove that they are 2-competitive. ii) We show that the optimal online time efficiency ratio is also 2. iii) We compare the performance of the proposed online algorithms in terms of transmitted data and/or completion time.

To the best of our knowledge, the mentioned problem (completion time minimization using online algorithms with an arrival data process) has not been considered before even in a discrete setup while we consider the continuous model for both harvested energy and data arrival.
The continuity assumption together with the arrival data existence makes our proof techniques to be different from the existing works.
Although most of existing research works assume a discrete model for the harvested energy to make the analysis tractable, a continuous model is more accurate for the EH sources in many applications \cite{varan2014energy,ottman2002adaptive}. Besides, as we show in Section~\ref{sec:num}, the discretizing of harvested energy curve reduces the efficiency of system. In addition, the motivation of continuous model for data arrival comes from the rateless codes, network calculus and using relays in communications \cite{milad}.
\section{System Model and Problem Description}
Assume that we have a single-user communication system, where the Tx is a node which harvests energy from a renewable source in a continuous fashion. Also, assume that the receiver (Rx) has enough energy to provide sufficient power for decoding at any rate that can be achieved by the Tx. Also, we have the following assumptions. The transmitted data curve ($B(t)$) and the transmitted energy curve ($E(t)$) denote the amount of transmitted data and utilized energy at the Tx in the interval $[0,t]$ for $t \in [0,\infty)$, respectively. $B(t)$ and $E(t)$ are continuous and they are differentiable, except probably in a finite number of points. The transmitted power curve, $p(t)$, denotes the amount of instantaneous used power at the Tx for $t \in [0,\infty)$, which is a piecewise continuous function. The arrival data curve ($B_{s}(t)$) and the harvested energy curve ($E_{s}(t)$) denote the amount of arrived data and harvested energy at the Tx in the interval $[0,t]$, respectively. $B_{s}(t)$ and $E_{s}(t)$ are bounded and they are differentiable, except probably in a finite number of points (in these points, $B_s(t)$ and $E_s(t)$ can have discontinuity or unequal right and left derivatives). To include the discrete case in our model $B_{s}(t)$ and $E_{s}(t)$ are assumed to be piecewise continuous. Moreover, the derivatives of $B_{s}(t)$ and $E_{s}(t)$ are bounded (except probably in a finite number of points) and piecewise continuous.
In addition, $r(p)$ is a continuous channel capacity function in which $r(0)=0$, ii) $r(p)$ is a non-negative strictly concave function in $p$, iii) $r(p)$ is differentiable, iv) $r(p)$ increases monotonically in $p$, and v) $\lim_{p \to\infty }r(p)=\infty$.

Before proceeding to the online algorithm, we define the problem at hand, i.e., the minimum time in which the Tx can transmit $B_{0}$ bits to the Rx. Note that we assume that until time $T_{B_{0}}\leq T_{off}$ in the Tx $B_{s}(T_{B_{0}})=B_{0}$, and we have no new arrival data after $T_{B_{0}}$. Our problem is defined in the following optimization problem:

\begin{small}
\begin{align}
    T_{\textrm{off}}=&\min_{p(t)}~~T\label{2}\\
    s.t.&~ B_{0}=\int_{0}^{T}r(p(t))dt\label{3}\\
    &\int_{0}^{t}p({t}')d{t}'\leq E_{s}(t),~ 0\leq t\leq T\label{4}\\
    &\int_{0}^{t}r(p({t}'))d{t}'\leq  B_{s}(t),~0\leq t\leq T, \label{5}
\end{align}\end{small}where $T_{\textrm{off}}$ is the minimum completion time in the offline algorithm. Equation \eqref{3} shows that until time $T$ should be sent $B_{0}$ bits; the inequality \eqref{4} shows that energy cannot be used while it still has not arrived (energy causality); the inequality \eqref{5} shows data cannot be sent while it still has not arrived (data causality).
We remark that the only works that considered online algorithms in a similar setup are \cite{vaze2013competitive,vaze2013dynamic,zheng2016online} which assumed the stored data in the beginning of transmission and their model were discrete (compared to our \emph{arrival data process} \eqref{5} and \emph{continuous} model, respectively). Thus, we exploit different approaches to prove our results. Also, we propose two algorithms and compare the performance of them in Theorem \ref{performace}. Moreover, our results hold for any channel model (with rate shown by a concave function $r(p)$). However, the results of \cite{vaze2013competitive,vaze2013dynamic} rely on the $\log$-type rate functions (Gaussian channel).

\section{Main Results}\label{sec:online}
\textbf{Notations:}
 Let $\varphi$ be an optimization problem which depends on a set of curves $\gamma$ and let $A$ be an online algorithm that works without knowing the future of set of curves $\gamma$. Now define $C_{A}^{\varphi}$ as the cost function of online algorithm $A$ for optimization problem $\varphi$, and $C_{O}^{\varphi}$ as cost function of optimal \emph{offline} algorithm $O$ for optimization problem $\varphi$. We say $A$ is a $\rho_{A}$-competitive online algorithm, if $\max_{\gamma } \frac{C_{A}^{\varphi}}{C_{O}^{\varphi}}\leq  \rho_{A}$.
  In this paper, we assume that $C_{A}^{\varphi}$ is the completion time of the transmission in $A$. Also, the subscript (on) refers to the online algorithm and (off) to the optimal offline algorithm.
Also, 
let $E_{rem}(t)$ be the amount of energy that is available in energy buffer at instant $t$. Also, let $B_{rem}(t)$ be the amount of the remaining data that should be transmitted at instant $t$, in order to transmit total $B_{0}$ bits.

In the following, we propose two algorithms for the optimization problem \eqref{2}-\eqref{5} where we assume that the Tx has only causal information about two curves $B_{s}(t)$ and $E_{s}(t)$. In our model, neither the distribution of the arrival data process nor the distribution of the harvested energy process are known. In both algorithms, the transmission process does not necessarily start immediately after reception of the first data and/or energy. In the second algorithm, the Tx never becomes silent after starting transmission, however, in the first algorithm, the transmission-silent cycles are repeated depending on the arrival data.\\
\textbf{Algorithm 1:}

\emph{1. Waiting phase:}  Wait till instant $T_{s_{1}}$ defined as:
\begin{align}
&T_{s_{1}}=\min ~~~~~~ t \\
s.t.&~~ B_0\leq  \lim\limits_{T\to \infty} T\times  r(\frac{E_{s}(t)}{T}),\quad
0< B_{s}(t)\label{20}.
\end{align}
The first term in \eqref{20} states that the Tx must wait until time $t$ having enough energy for transmitting $B_{0}$ bits (such that we can transmit $B_{0}$ bits in an interval with finite length) and the second term in \eqref{20} is due to data causality condition.

\emph{2. Transmission phase:} For $t\geq T_{s_{1}}$, the transmitted power curve (denoted as $p_{on}(t)$) satisfies the following equality,
{\small\begin{align}
\dfrac{\overbrace{E_{s}(t)-\int_{T_{s}}^{t}p_{on}(t^{'})dt^{'}}^{E_{rem}(t)}}{p_{on}(t)}r(p_{on}(t))=\overbrace{B_{0}-\int_{T_{s}}^{t}r(p_{on}(t^{'}))dt^{'}}^{B_{rem}(t)}\label{11},
\end{align}}if $B_{s}(t)-\int_{T_{s}}^{t}r(p_{on}(t^{'}))dt^{'}>0$. Note that, depending on the arrival data curve, the Tx may be silent (in some intervals of the transmission phase), since there is no data to transmit in these intervals.
\eqref{11} describes $p_{on}(t)$ in each instant $t$ by making a compromise among the remaining energy ($E_{rem}(t)$) and remaining data ($B_{rem}(t)$) in each instant $t$. The idea comes from that in each instant $t$, if there is no received energy in the future, the remaining data will be transmitted optimally with a fixed transmission power (in a minimum completion time).\\
\textbf{Algorithm 2:}

\emph{1. Waiting phase:} Wait till instant $T_{s_{2}}$ defined as:
\begin{align}
& T_{s_{2}}=\min ~~~~~~ t \\
s.t.&~~  t \times r(\frac{E_{s}(t)}{t})\geq B_{0}, \quad B_{s}(t)= B_{0}\label{10}.
\end{align}
The first term in \eqref{10} states that the Tx must wait until time $t$ having enough energy for transmitting $B_{0}$ bits and the second term in \eqref{10} forces the Tx to wait until all $B_{0}$ bits data is received.

\emph{2. Transmission phase:} For $t\geq T_{s_{2}}$, the transmitted power curve (denoted as $p_{on}(t)$) satisfies \eqref{11}.

We start with analysis of online time efficiency ratio for the Algorithm 2 (which is simpler to analyze). Then, we show that the Algorithm 1 is more efficient in Theorem \ref{performace}.

To compare with existing works, we remark that the resulted transmitted curves of Algorithm 1 is different from the algorithms in \cite{vaze2013competitive} and \cite{zheng2016online} in two aspects: 1) In our proposed Algorithm~1, due to the data arrival process, we add another condition to keep Tx silent until we have data to transmit (the second term of \eqref{20}). However, the Generalized Lazy Online Algorithm (GLO) in \cite{vaze2013competitive} and algorithm 2 in \cite{zheng2016online} do not consider such conditions. 2) In the proposed online Algorithm 1, the Tx is silent in some intervals of Transmission phase (due to data arrival process), while in the algorithms of \cite{vaze2013competitive,zheng2016online}, the Tx never becomes silent after transmission starts.
It is worth noting that these two differences in the algorithm as well as the continuity assumption make our proofs to be different from the existing results on online algorithms.
In the following, we use the subscribes 1 and 2 for the variables corresponding to algorithms 1 and 2, respectively.
\begin{lemma}\label{IV.1}
 $p_{on_{2}}(t)$ is nondecreasing.
\end{lemma}
\begin{proof} We use contradiction: we assume that for $t\in [t_{0},t_{0}+\delta]$, $p_{on_{2}}(t)$ is monotonically decreasing. For $t=t_{0}+\delta$, \eqref{11} implies:
\begin{small}
 \begin{align}
\dfrac{E_{s}(t_{0}+\delta)-\int_{T_{s_{2}}}^{t_{0}+\delta}p_{on_{2}}(t)dt}{p_{on_{2}}(t_{0}+\delta)}&r(p_{on_{2}}(t_{0}+\delta))=\\
B_{0}-\int_{T_{s_{2}}}^{t_{0}}r(p_{on_{2}}(t))dt-&\int_{t_{0}}^{t_{0}+\delta}r(p_{on_{2}}(t))dt.
\end{align}\end{small}

Noting that $\frac{r(p)}{p}$ is monotonically decreasing and $E_{s}(t_{0})\leq E_{s}(t_{0}+\delta)$, we get
\begin{small}
\begin{align*}
\dfrac{E_{s}(t_{0})-\int_{T_{s_{2}}}^{t_{0}}p_{on_{2}}(t)dt-\int_{t_{0}}^{t_{0}+\delta}p_{on_{2}}(t)dt}{p_{on_{2}}(t_{0})}r(p_{on_{2}}(t_{0}))\\
<\dfrac{E_{s}(t_{0})-\int_{T_{s_{2}}}^{t_{0}}p_{on_{2}}(t)dt}{p_{on_{2}}(t_{0})}r(p_{on_{2}}(t_{0}))-\int_{t_{0}}^{t_{0}+\delta}r(p_{on_{2}}(t))dt.
\end{align*}\end{small}Hence,
\begin{small}
\begin{align}
\dfrac{\int_{t_{0}}^{t_{0}+\delta}r(p_{on_{2}}(t))dt}{\int_{t_{0}}^{t_{0}+\delta}p_{on_{2}}(t)dt}<\dfrac{r(p_{on_{2}}(t_{0}))}{p_{on_{2}}(t_{0})}.\label{eqn:algo}
 \end{align}\end{small}Now applying the Riemann sum, we have
 \begin{small}
\begin{align}
\dfrac{\int_{t_{0}}^{t_{0}+\delta}r(p_{on_{2}}(t))dt}{\int_{t_{0}}^{t_{0}+\delta}p_{on_{2}}(t)dt}=\dfrac{\lim_{N \to\infty}\Delta\sum_{k=0}^{N-1}r(p_{on_{2}}(t_{0}+k\Delta))}{\lim_{N \to\infty}\Delta\sum_{k=0}^{N-1}p_{on_{2}}(t_{0}+k\Delta)},\label{eqn:Rie}
\end{align}
\end{small}
where $\Delta=\frac{\delta}{N}$. Also, since $p_{on_{2}}(t)$ is monotonically decreasing in $[t_{0},t_{0}+\delta]$ we have,
\begin{small}
\begin{align*}
 \dfrac{r(p_{on_{2}}(t_{0}))}{p_{on_{2}}(t_{0})}<\dfrac{r(p_{on_{2}}(t_{0}+\Delta))}{p_{on_{2}}(t_{0}+\Delta)}<...<\dfrac{r(p_{on_{2}}(t_{0}+(N-1)\Delta))}{p_{on_{2}}(t_{0}+(N-1)\Delta)}.
\end{align*}
\end{small}
Thus for $N\in [2,\infty]$,
\begin{small}
\begin{align}
\dfrac{r(p_{on_{2}}(t_{0}))}{p_{on_{2}}(t_{0})}< \dfrac{\Delta\sum_{k=0}^{N-1}r(p_{on_{2}}(t_{0}+k\Delta))}{\Delta\sum_{k=0}^{N-1}p_{on_{2}}(t_{0}+k\Delta)}.\label{eqn:delta}
\end{align}\end{small} Combining \eqref{eqn:Rie} and \eqref{eqn:delta} results in:
\begin{small}
\begin{align}
\dfrac{r(p_{on_{2}}(t_{0}))}{p_{on_{2}}(t_{0})}<\dfrac{\int_{t_{0}}^{t_{0}+\delta}r(p_{on_{2}}(t))dt}{\int_{t_{0}}^{t_{0}+\delta}p_{on_{2}}(t)dt}.\label{eqn:delta2}
 \end{align}\end{small}\eqref{eqn:algo} and \eqref{eqn:delta2} result in a contradiction and complete proof.
 \end{proof}
\begin{lemma}\label{IV.2}$T_{s_{2}} \leq T_{off}$.
\end{lemma}
\begin{proof}
The result can be easily proved using \cite[Lemma 2]{milad} and thus it is omitted for brevity.
\end{proof}
Now, we show that Algorithm 2, without any information about distribution of two process $B_{s}(t)$ and $E_{s}(t)$, can transmit $B_0$ bits in less than twice of completion time in the optimal offline algorithm.
\begin{theorem}\label{IV.4}$T_{on_{2}} \leq 2~T_{off}$.
\end{theorem}
\begin{proof}
Since $p_{on_{2}}(t)$ is nondecreasing (Lemma \ref{IV.1}), it is enough to show that when we let $p_{on_{2}}(t)=p_{on_{2}}(T_{s_{2}})$ for all instants $t\geq T_{s_{2}}$, the claim is true. Obviously, the algorithm with $p_{on_{2}}(t)\geq p_{on}(T_{s_{2}})$ transmits $B_0$ bits in less time (due to the fact that $r(p)$ is an increasing function). Thus, let $p_{on_{2}}(t)=p_{on_{2}}(T_{s})$ for all instants $t$. The instant ($t_{c}$) in which all of $B_{0}$ bits are transmitted, is obtained from the following:
{\small$(t_{c}-T_{s_{2}})~r(\dfrac{E_{s}(T_{s_{2}})}{t_{c}-T_{s_{2}}})=B_{0},$} which accompanying the first term of \eqref{10} results in $t_{c}-T_{s_{2}}\leq T_{s_{2}}$. Now, applying Lemma \ref{IV.2}, we have $t_{c}<2T_{s_{2}}<2T_{off}$. This completes the proof.
\end{proof}
Now, we obtain the optimal online time efficiency ratio.
\begin{theorem}\label{optimal2cop}
The optimal online time efficiency ratio is equal to $2$, i.e., $\min\{\rho_{A}\}= 2$ when $A$ is the optimal online algorithm.
\end{theorem}
\begin{proof}
Based on the definition of $\rho_{A}$, the optimal online algorithm $A$ is $\rho_{A}$-competitive if its optimal online time efficiency ratio is lower than $\rho_{A}$ for \emph{any} channel capacity function $r(p)$ with the specifications stated in Section~II and any harvested energy/arrival data processes. In Theorem \ref{IV.4}, we showed that the optimal online algorithm is $2$-competitive by proposing an online algorithm with $\rho_{A}\leq 2$ for any channel capacity function $r(p)$. This completes the upper bound proof.

Now we exploit the technique of \cite{zheng2016online} to find a lower bound on the optimal online time efficiency ratio.
Assuming all $B_{0}$ amount of data has been stored at the beginning of transmission, a two-person zero-sum game has been considered in \cite[Section VI]{zheng2016online}, between the Tx (the strategy designer) and nature where the kernel function of the game is the online time efficiency ratio.
In \cite[Section VI]{zheng2016online}, the authors shows that this game has an $\epsilon$ saddle point for any positive $\epsilon$ and this two-person zero-sum game has a value of $2$ in pure strategy.
To show the lower bound it is enough to provide an example where no online algorithm $B$ has $\rho_B<2$ (the worst-case approach).
Because the system model in \cite{zheng2016online} is an special case of our model (with no arrival data and considering discrete harvested energy arrival), the above example also works for our model. This completes the lower bound proof.
\end{proof}
In theorems \ref{IV.4} and \ref{optimal2cop}, we have shown that Algorithm 2 is optimal online algorithm in sense of competitive-ratio. In fact, there are a lot of other algorithms that are 2-competitive. In the competitive-ratio analysis, we only consider the worst case, however in practice we often wish to minimize the completion time in a specific case (not necessarily the worst case).
In the next theorem, we prove that Algorithm 2 never performs better than Algorithm 1. Also, it can be easily shown that Algorithm 2 performs better than the extension of algorithm 1 in \cite{zheng2016online} to the continuous model by adding another condition to keep the Tx silent until the Tx obtains $B_{0}$ data. The reason is that in Algorithm 2 in the transmission phase the transmission power is increasing if there is a new energy arrival, while in the extension of algorithm 1 in \cite{zheng2016online} the transmission power is constant.
\begin{theorem}\label{performace}
$B_{on_{2}}(t)\leq B_{on_{1}}(t)$.
\end{theorem}
\begin{figure}\label{fig1}
  \centering
  \includegraphics[width=.45\textwidth,left]{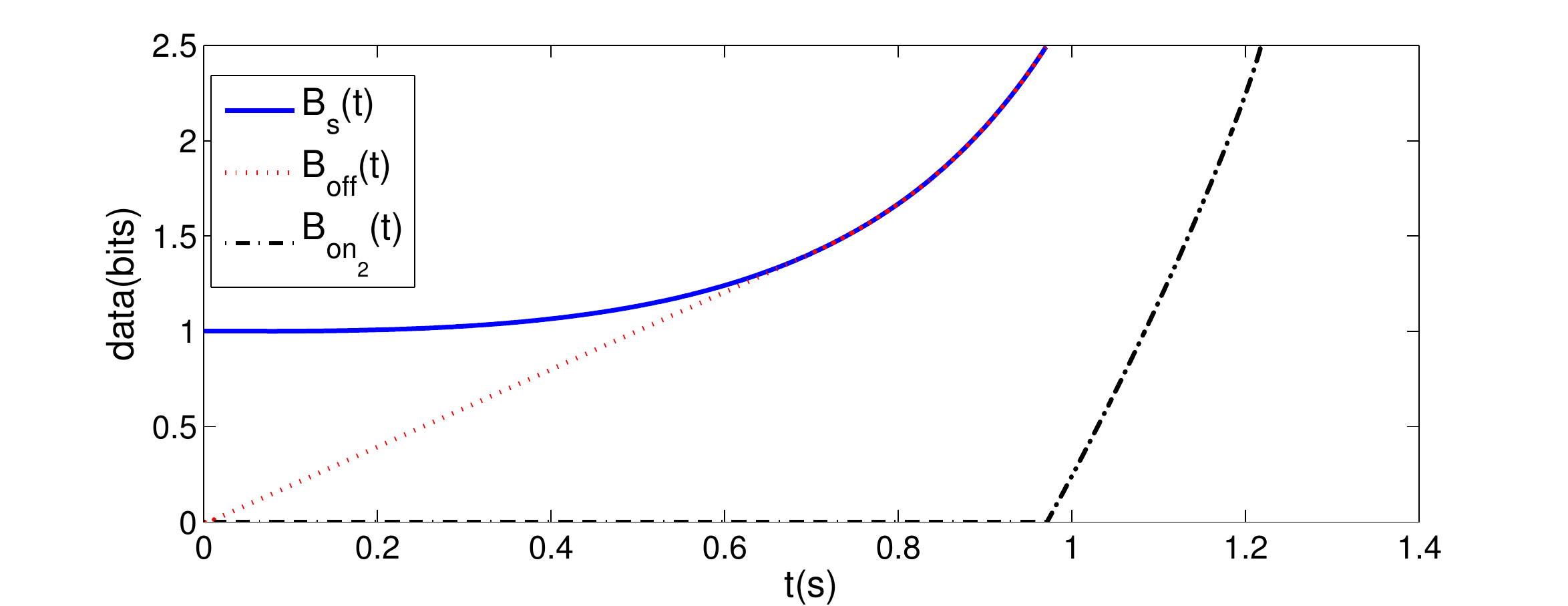}
  \caption{Online and Optimal offline algorithms without discretizing }
  \end{figure}
\begin{proof}
Assume that $0=t_{0}<t_{1}<t_{2}, ..., <t_{2n+1}<t_{2n+2}=T_{on_{1}}$ are all instants, in which the Tx switches from silence to data transmission or vice versa in Algorithm 1. Assume that in $[0,t_1)$, the Tx is silent. If $[0,t_1)$ is not the first interval in which Tx is silent, or there is no interval where the Tx is silent, the proof is similar. We denote $B_{on_{1}}(t)$ as $B_{i}(t)$ in the $i$-th data transmission interval $(t_{2i-1},t_{2i})$. In the other intervals, the Tx is silent and transmitted data curve is constant. We denote $B_{new}(t)$ as a new transmitted data curve such that there is only one silent interval and then one data transmission interval. Its transmission interval is formed by merging all transmission intervals of $B_{on_{1}}(t)$ (shifted to the right). In fact,
{\small
\begin{align*}
B_{new}(t)=\left\{\begin{matrix}
0 & 0\leq t< a\\
B_{on_{1}}(t-(x_{i})+t_{2i-1}) &x_{i} <t<x_{i}+t_{2i}-t_{2i-1}
\end{matrix}\right.,
\end{align*}}where $a=\sum\limits_{m=0}^{n}(t_{2m+1}-t_{2m})$, $x_{i}=\sum\limits_{m=0}^{n}(t_{2m+1}-t_{2m})+\sum\limits_{j=1}^{i}(t_{2j-2}-t_{2j-3})$, $t_{0}=t_{-1}=0$, $i\in \mathbb{N}$ and $i\in [1,n+1]$.
Intuitively, we shift curves $B_{i}(t)$ to the right to eliminate all intervals, in which $B_{on_{1}}(t)$ is constant. Hence, $B_{new}(t)=0$ in $(0,a)$ and based on Lemma \ref{IV.1} $B_{new}(t)$ is convex in $(a, T_{on_{1}})$. Since $B_{on_{1}}(t)$ is non-decreasing, we have $B_{new}(t)\leq B_{on_{1}}(t), \forall t\in [0,T_{on_{1}}]$. Now we prove that $B_{on_{2}}(t)\leq B_{new}(t)$, $\forall t\in[0,T_{on_1}]$. \eqref{20} and \eqref{10} result in $T_{s_{1}}\leq T_{s_{2}}$. If $T_{s_{1}}=T_{s_{2}}$, then $B_{on_{2}}(t)=B_{new}(t), \forall t\in [0, T_{on_{1}}]$. Now, we consider $T_{s_{1}}< T_{s_{2}}$ and use contradiction. Thus, there exists $t\in (T_{s_{2}},T_{on_{1}})$, where $B_{on_{2}}(t)> B_{new}(t)$. Assume that $a$ is the first point, for which there exists $b>a$ such that $\forall t\in(a,b)$, $B_{on_{2}}(t)> B_{new}(t)$ and $B_{on_{2}}(a)=B_{new}(a)$. Hence, we obtain $p_{on_{2}}(a)\geq p_{new}(a)$, because $p(t)$ is piecewise continuous in both algorithms. Since $B_{new}(t)$ is convex, increasing, $B_{on_{2}}(0)= B_{new}(0)$, $B_{on_{2}}(a)= B_{new}(a)$, and $B_{on_{2}}(t)<B_{new}(t)~\forall t\in(0,a)$, we have $E_{new}(a)< E_{on_{2}}(a)$ based on \cite[Lemma 23]{milad}. Now, noting $B_{on_{2}}(a)= B_{new}(a)$ and \eqref{11}, we have,

{\small\begin{align}\label{contradict}
    \frac{r(p_{new}(a))}{p_{new}(a)}(E_{s}(a)-E_{new}(a))=\frac{r(p_{on_{2}}(a))}{p_{on_{2}}(a)}(E_{s}(a)-E_{on_{2}}(a)).
    \end{align}} Combining $E_{new}(a)< E_{on_{2}}(a)$ and \eqref{contradict}, we obtain $\frac{r(p_{new}(a))}{p_{new}(a)}<\frac{r(p_{on_{2}}(a))}{p_{on_{2}}(a)}$. Thus, $p_{new}(a)>p_{on_{2}}(a)$ holds since $\frac{r(p)}{p}$ is monotonically decreasing. This is a contradiction.
\end{proof}

\section{Numerical Results and Concluding Remarks}\label{sec:num}
  \begin{figure}\label{fig2}
    \centering
    \includegraphics[width=.45\textwidth,right]{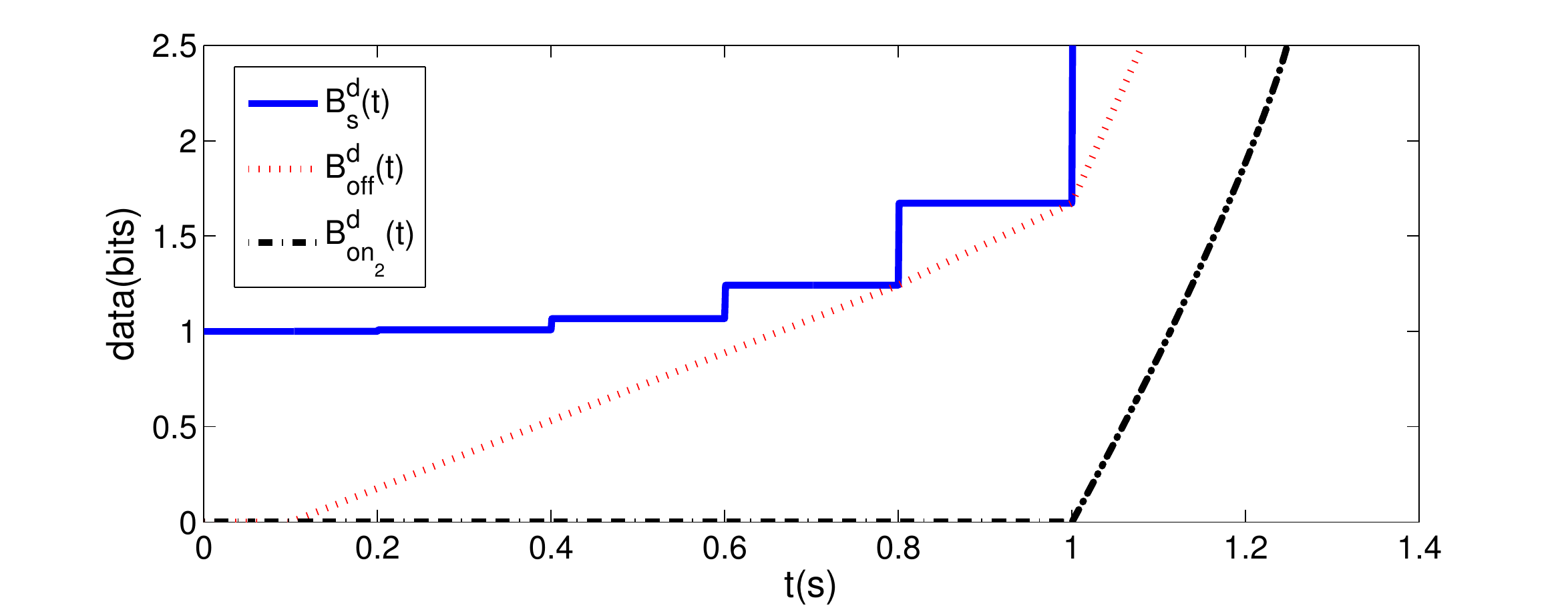}
    \caption{Online and Optimal offline algorithms with discretizing }
    \end{figure}
In this section, we provide a numerical example to explain our results for Algorithm 2. Consider an additive white Gaussian noise channel with a limited bandwidth $W=1$ Hz. Assuming $\frac{\rm{channel~gain}}{\rm{noise~ power}\times W}=1$, we have $r(p)=\log(1+p)$, where the $\log$ is in base $2$. Assume that $E_{s}(t)=100 t^{2}$ J, $B_{s}(t)=e^{t^{3}}$ bits, and $B_{0}=2.5$ bits. Fig. 1 illustrates $B_{s}(t)$, $B_{off}(t)$ and $B_{on_{2}}(t)$, and it can be seen that $p_{on_{2}}(t)$ is nondecreasing (Lemma \ref{IV.1}) and $T_{s}=T_{off}$ (Lemma \ref{IV.2}). Also, observing $T_{off}=.97s$ and $T_{on_{2}}=1.21s$ results in $\frac{T_{on_{2}}}{T_{off}}=1.24\leq 2$ (Theorem \ref{IV.4}). Fig.~2 illustrates the result for discretized model of $E_{s}(t)$ and $B_{s}(t)$, where $B_{s}^{d}(t)$ is the discretized version of $B_{s}(t)$. It can be easily seen that $T_{off}^{d}>T_{off}$ and $T_{on_{2}}^{d}>T_{on_{2}}$. Hence, As mentioned in Introduction the discretizing have reduced the efficiency.

To conclude, in this paper, we assumed an EH system with continuous arrival data and continuous harvested energy curves, the obtained results are held even for the discrete model. However, the most of research works in this area consider a discrete model because of mathematical tractability of the ensuing system optimization. We proposed two online algorithms which achieve the upper bound of $2$ on the online time efficiency ratio. Also, we showed that this ratio is $2$ for the optimal online algorithm. Moreover, we compared the performance of proposed online algorithms.

\bibliographystyle{IEEEtranTCOM}
\bibliography{IEEEabrv}

\begin{thebibliography}{1}
\baselineskip 12pt
\providecommand{\url}[1]{#1}
\csname url@samestyle\endcsname
\providecommand{\newblock}{\relax}
\providecommand{\bibinfo}[2]{#2}
\providecommand{\BIBentrySTDinterwordspacing}{\spaceskip=0pt\relax}
\providecommand{\BIBentryALTinterwordstretchfactor}{4}
\providecommand{\BIBentryALTinterwordspacing}{\spaceskip=\fontdimen2\font plus
\BIBentryALTinterwordstretchfactor\fontdimen3\font minus
  \fontdimen4\font\relax}
\providecommand{\BIBforeignlanguage}[2]{{%
\expandafter\ifx\csname l@#1\endcsname\relax
\typeout{** WARNING: IEEEtran.bst: No hyphenation pattern has been}%
\typeout{** loaded for the language `#1'. Using the pattern for}%
\typeout{** the default language instead.}%
\else
\language=\csname l@#1\endcsname
\fi
#2}}
\providecommand{\BIBdecl}{\relax}
\BIBdecl

\bibitem{tutuncuoglu2012optimum}
K.~Tutuncuoglu and A.~Yener, ``Optimum transmission policies for battery
  limited energy harvesting nodes,'' \emph{IEEE Trans. Wireless Commun.},
  vol.~11, no.~3, 2012.

\bibitem{vaze2013dynamic}
R.~Vaze, R.~Garg, and N.~Pathak, ``Dynamic power allocation for maximizing
  throughput in energy-harvesting communication system,'' \emph{IEEE/ACM Trans.
  Netw.}, vol.~22, no.~5, 2014.

\bibitem{varan2014energy}
B.~Varan, K.~Tutuncuoglu, and A.~Yener, ``Energy harvesting communications with
  continuous energy arrivals,'' in \emph{IEEE ITA}, 2014.

\bibitem{milad}
M.~Rezaee, M.~Mirmohseni, and M.~R. Aref, ``Energy harvesting systems with
  continuous energy and data arrivals: the optimal offline and a heuristic
  online algorithms,'' accepted in IEEE J. Sel. areas Commun., 2016.

\bibitem{rezaee2016optimal}
M.~Rezaee, M.~Mirmohseni, V.~Aggarwal, and M.~R. Aref, ``Optimal transmission
  policies for multi-hop energy harvesting systems,'' \emph{arXiv preprint
  arXiv:1612.09496}, 2016.

\bibitem{yang2012optimal}
J.~Yang and S.~Ulukus, ``Optimal packet scheduling in an energy harvesting
  communication system,'' \emph{IEEE Trans. Commun}, vol.~60, no.~1, 2012.

\bibitem{vaze2013competitive}
R.~Vaze, ``Competitive ratio analysis of online algorithms to minimize packet
  transmission time in energy harvesting communication system,'' in \emph{IEEE
  INFOCOM}, 2013.

\bibitem{zheng2016online}
X.~Zheng, S.~Zhou, and Z.~Niu, ``On the online minimization of completion time
  in an energy harvesting system,'' in \emph{IEEE WiOpt}, 2016.

\bibitem{ottman2002adaptive}
G.~K. Ottman, H.~F. Hofmann, A.~C. Bhatt, G.~Lesieutre \emph{et~al.},
  ``Adaptive piezoelectric energy harvesting circuit for wireless remote power
  supply,'' \emph{IEEE Trans. Power Electron.}, vol.~17, no.~5, 2002.

\end{thebibliography}
%




\end{document}